\documentclass[journal, 10pt]{IEEEtran}

\usepackage{multirow}
\usepackage[font=scriptsize,caption=false,labelsep=space]{subfig}
\usepackage{mathrsfs}
\usepackage{graphicx,float,wrapfig,epstopdf,amsmath}
\usepackage[]{algorithmicx}
\usepackage{algpseudocode,algorithm}
\usepackage{balance}
\usepackage{enumitem}
\usepackage{cuted}

\usepackage{mathtools,lipsum}
\usepackage{amsmath}
\usepackage{amssymb}
\usepackage{amsthm}
\usepackage{graphicx}
\usepackage{epstopdf}
\usepackage{times}
\usepackage{textcomp,cite}

\usepackage{url}
\usepackage{hyperref}

\usepackage{multirow}
\usepackage{threeparttable}
\graphicspath{{./figures/}}

\DeclareMathOperator*{\minimize}{minimize} 
\DeclareMathOperator*{\subjectto}{s.\hspace{3pt} t.:\hspace{3pt}} 

\usepackage[table]{xcolor}
\definecolor{infocolor}{RGB}{213,229,255}
\definecolor{inteins}{RGB}{128,179,255}
\definecolor{color1}{RGB}{199,209,232}
\definecolor{color2}{RGB}{230,231,233}

\usepackage[table]{xcolor}
\newtheorem{theorem}{Theorem}

\newtheorem{lemma}[theorem]{Lemma}

\usepackage{setspace}

\begin{document}

	\title{Terahertz-Band Direction Finding With Beam-Split and Mutual Coupling Calibration }

	\author{\IEEEauthorblockN{Ahmet M. Elbir, \textit{Senior Member, IEEE},
			Kumar Vijay Mishra,
			\textit{Senior Member, IEEE},  \\
			and		Symeon Chatzinotas, \textit{Fellow, IEEE}
		}
		\thanks{This work was supported in part by the Horizon Project TERRAMETA.}
		\thanks{A. M. Elbir is with the  University of Luxembourg, Luxembourg; and Duzce University, Turkey (e-mail: ahmetmelbir@ieee.org).} 
		\thanks{K. V. Mishra is with the United States Army Research Laboratory, Adelphi, MD 20783 USA; and the University of Luxembourg, Luxembourg  (e-mail: kvm@ieee.org).}
		\thanks{S. Chatzinotas is with the University of Luxembourg, Luxembourg (e-mail: symeon.chatzinotas@uni.lu). }

	}
	\maketitle
	
	\begin{abstract}
		
		Terahertz (THz) band is currently envisioned as the key building block to achieving the future sixth generation (6G) wireless systems. The ultra-wide bandwidth and very narrow beamwidth of THz systems offer the next order of magnitude in user densities and multi-functional behavior. However, wide bandwidth results in a frequency-dependent beampattern causing the beams generated at different subcarriers split and point to different directions. Furthermore, mutual coupling degrades the system's performance. This paper studies the compensation of both beam-split and mutual coupling for direction-of-arrival (DoA) estimation by modeling the beam-split and mutual coupling as an array imperfection. We propose a subspace-based approach using multiple signal classification with CalibRated for bEAam-split and Mutual coupling (CREAM-MUSIC) algorithm for this purpose. Via numerical simulations, we show  the proposed CREAM-MUSIC approach accurately estimates the DoAs in the presence of beam-split and mutual coupling.

	\end{abstract}
	%
	
	
	%
	
	\vspace{-12pt}
	\section{Introduction}
	\label{sec:Introduciton}
	Terahertz (THz) band ($0.1$-$10$ THz) is widely viewed as a technology to enable significant performance enhancements in sixth-generation (6G) wireless networks~\cite{thz_Akyildiz2022May}. An accuracy of milli-degree in direction-of-arrival (DoA) estimation  is critical to guaranteeing the reliability of THz sensing as well as communications~\cite{milliDegree_doa_THz_Chen2021Aug,elbir2022Aug_THz_ISAC}. However, significant challenges remain in realizing this goal at the THz band owing to high path losses, complex propagation/scattering phenomena, and the use of extremely large arrays~\cite{ummimoTareqOverview,thz_Akyildiz2022May,thz_metasurface_Venkatesh2022Jun}. 
	
	In particular, the ultra-dense integration of the antennas in the THz array is accompanied by mutual coupling (MC) which leads to degradation in the direction finding accuracy \cite{milliDegree_doa_THz_Chen2021Aug,thz_DOA_Nayir2022Oct,mutualCoupling_THZ_Zhang2019Mar}.	 
	The existing works on reducing MC in THz mostly involve array design with new metamaterials~\cite{mutualCoupling_THz2_Jafargholi2018Oct}, graphene monolayers~\cite{mutualCoupling_THz_graphene_Moreno2015Dec}, THz metasurfaces~\cite{thz_metasurfaceAlibakhshikenari2018} or frequency-selective graphene surfaces~\cite{mutualCoupling_THZ_Zhang2019Mar}, while MC calibration at THz-band via signal processing remains relatively unexamined. 
	Furthermore, THz arrays also suffer from \textit{beam-split} arising from the subcarrier-independent (SI) analog beamformers (ABs)~\cite{thz_beamSplit,elbir_THZ_CE_ArrayPerturbation_Elbir2022Aug,elbir2022Nov_SPM_beamforming}. This causes the generated beams at different subcarriers split and point to different directions. {At lower frequencies (e.g., up to millimeter-wave~\cite{spatialWideband_Wang2018May}), \textit{beam-squint} is used to describe the same phenomenon, and negative group-delay networks are used to compensate for it \cite{beam_squinting_Mirzaei2015Feb}. Whereas beam-squint causes slight deviations in the beam direction while they still cover the targets with their mainlobes, the generated beams are completely non-overlapping in \textit{beam-split}. The latter is, therefore, a more severe form of the former (Fig.~\ref{fig_MUSIC_spectra}a)}. 	This angular deviation eventually impacts MC, which is direction-dependent because of the directional beampattern~\cite{elbir_directionDependent_MC_Elbir2017Jan}. The existing techniques to mitigate beam-split are mostly hardware-based~\cite{thz_beamSplit}. Specifically, additional hardware components such as time-delayer networks to realize subcarrier-dependent (SD) ABs, they are expensive because each phase shifter is connected to multiple delayer elements. Further, each such element consumes approximately $100$ mW which is more than that of a single phase shifter ($40$ mW) at THz~\cite{elbir2022Aug_THz_ISAC}. 
	

%
	
	Although THz channel estimation~\cite{elbir_THZ_CE_ArrayPerturbation_Elbir2022Aug} and hybrid analog/digital beamforming~\cite{thz_beamSplit,elbir2021JointRadarComm} under beam-split have been explored in prior THz studies, these algorithms did not examine either DoA estimation or MC calibration. The hybrid architectures at mm-Wave~\cite{doaEst_mmWave_AngleDomain_Fan2017Dec,doaEst_mmWave_Zhang2021Oct} and THz~\cite{milliDegree_doa_THz_Chen2021Aug} so far employ orthogonal matching pursuit (OMP)~\cite{doaEst_mmWave_AngleDomain_Fan2017Dec},  maximum likelihood (ML)~\cite{doaEst_mmWave_Zhang2021Oct} and \textit{MU}ltiple \textit{SI}gnal \textit{C}lassification (MUSIC)~\cite{milliDegree_doa_THz_Chen2021Aug,thz_DOA_Nayir2022Oct} for DoA estimation but exclude beam-split and MC. The calibration of beam-split in THz systems is considered for different array geometries, e.g., uniform linear/rectangular array (ULA/URA)~\cite{elbir_THz_CE_BSAOMP_Elbir2023Feb,elbir_THZ_CE_ArrayPerturbation_Elbir2022Aug,elbir_THZ_HB_Aunified_TVT_Elbir2023Apr}. Specifically, the ULA provides a simple design whereas the URA-based THz systems allow a compact and efficient design with array/group of subarrays architecture~\cite{elbir2021JointRadarComm}. While the calibration procedure is the similar, the use of URA introduces the calibration of beam-split in both azimuth and elevation.      
	

	In this work, contrary to previous works, we address the aforementioned shortcomings by considering beam-split as an array imperfection in a similar way as MC has been investigated in the well-studied array calibration theory~\cite{friedlander,elbir_directionDependent_MC_Elbir2017Jan,array_calibration_Viberg2009Jan}. Specifically, we model the beam-split errors as a diagonal matrix, which represents a linear transformation between the nominal and actual steering vectors corrupted by beam-split. Using this transformation and incorporating the signal-noise subspace orthogonality, we develop a \textit{c}alib\textit{r}ated for b\textit{ea}am-split and \textit{M}C MUSIC (CREAM-MUSIC) algorithm to obtain	accurate DoA estimates. In particular, we introduce an alternating approach, wherein the DoAs, beam-split, and MC coefficients are iteratively estimated. For DoA, while the degree of beam-split is proportionally known \textit{a priori}, it depends on the unknown target direction. For example, consider $f_m$ and $f_M$ to be the frequencies for, respectively, the $m$-th and last/highest subcarriers. When $\theta$ is the physical target direction, then the spatial direction corresponding to the $m$-th-subcarrier is shifted by $\frac{f_m}{f_M}\theta$. 	We then construct the CREAM-MUSIC spectra which accounts for this deviation and thereby \textit{ipso facto} mitigates the effect of beam-split (Fig.~\ref{fig_MUSIC_spectra}b-c). 
	\begin{figure}[t]
		\centering
		{\includegraphics[draft=false,width=.7\columnwidth]{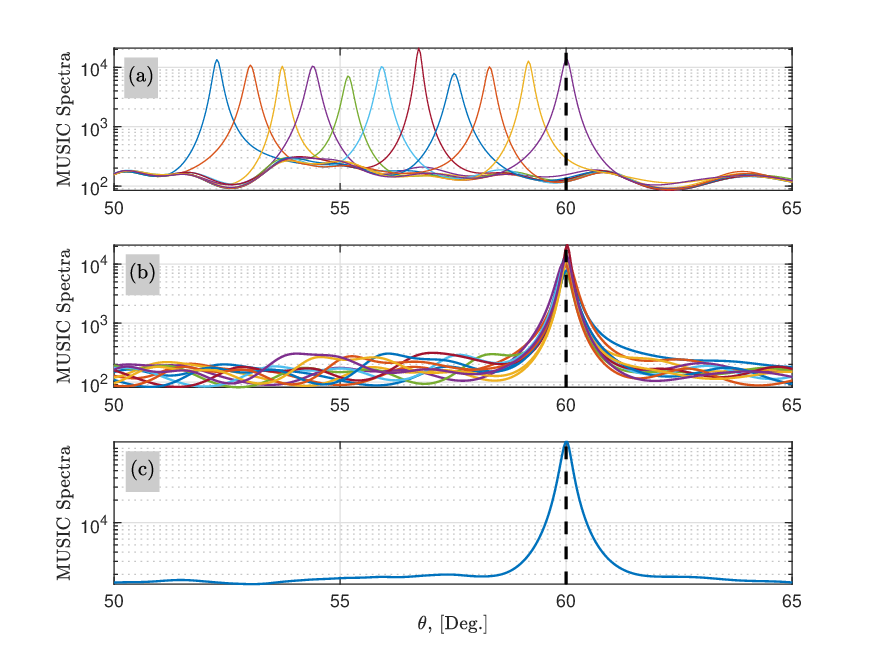} } 
		\caption{The MUSIC spectra of (a) all subcarriers without beam-split compensation, (b) all subcarriers with beam-split compensation, (c) CREAM-MUSIC algorithm for the source location depicted with vertical line. $f_c=300$ GHz, $B = 30$ GHz, $M = 11$, $\tilde{\theta}=60^\circ$, $N=128$ and $N_\mathrm{RF}=8$. 
		}
		\label{fig_MUSIC_spectra}
	\end{figure}

	\section{System Model \& Problem Formulation}
	\label{sec:SignalModel}
	Consider a wideband  THz ultra-massive multiple-input multiple-output uplink scenario, wherein the base-station (BS) employs hybrid analog/digital beamformers with $N$-element  ULA and $N_\mathrm{RF}$ radio-frequency (RF) chains. In  sensing applications, the equivalent signals would be from $K$ targets. Without loss of generality, assume the targets are in the far-field of the BS with the physical DoAs $\theta_k$ for $k\in \{1,\cdots, K\} = \mathcal{K}$, where  $\theta_{k} = \sin \tilde{\theta}_{k}$ ($\tilde{\theta}_{k}\in [-\frac{\pi}{2},\frac{\pi}{2}]$).

	To begin with, we first consider the conventional MC- and  beam-split- free scenario, wherein the $N\times 1$	steering vector corresponding to the $k$-th target is 
	\begin{align}
	\label{steeringVec_nominal}
	{\mathbf{a}}(\theta_{k}) = [1, e^{\mathrm{j} 2\pi \frac{d}{\lambda^\star} \theta_{k} },\cdots, e^{\mathrm{j} 2\pi (N-1)\frac{d}{\lambda^\star} \theta_{k} }]^\textsf{T},
	\end{align}
	where $d$ stands for  antenna spacing and $\lambda^\star  = \frac{\nu}{f_M}$ denotes the wavelength corresponding to the highest subcarrier frequency. Hence, $d$  is  selected as $d = \frac{\nu}{2\max_m f_m} = \frac{\nu}{2f_M}$ in order to avoid spatial aliasing~\cite{wideband_doaEst_Wideband_Friedlander1993Apr}. Here, $\nu$ is speed of light and $f_m = f_c + \frac{B}{M} (m - 1-\frac{M-1}{2})$ is the $m$-th subcarrier frequency ($ m \in  \{1,\cdots, M\}= \mathcal{M}$) with $f_c$ and $B$ being the carrier frequency and total bandwidth, respectively.
	
	In order to sense the targets, the BS employs the sensing signal $\mathbf{s}_m(t_i)\in \mathbb{C}^{N_\mathrm{RF}}$ and the SI precoder $\mathbf{F}\in \mathbb{C}^{N\times N_\mathrm{RF}}$. Then,  the BS activates $N_\mathrm{RF}$ RF chains and  applies $\mathbf{F}\in \mathbb{C}^{N\times N_\mathrm{RF}}$ to the sensing signal $\mathbf{s}_m(t_i)$, and the $N\times 1$ transmit signal becomes $	\mathbf{x}_{m}(t_i) = \mathbf{F}\mathbf{s}_m(t_i) $
	where $i = 1,\cdots, T$ and $T$ is the number of data snapshots.
	
	\textbf{Beam-Split:} 
	During the search phase of the THz radar, the AB is designed with the structure of steering vectors corresponding to the search directions. For an arbitrary  $\mathbf{a}(\theta_k)$  for direction $\theta_k$. Then, the $n$-th element of $\mathbf{a}(\theta_k)$ is defined as $[\mathbf{a}(\theta_k)]_n = 1/\sqrt{N} \exp \big\{\mathrm{j} (n-1) \frac{2\pi d }{\lambda_m} \theta_k   \big\}$, i.e.,
	\begin{align}
[\mathbf{a}(\theta_k)]_n = 1/\sqrt{N} \exp \big\{\mathrm{j} (n-1) \pi \eta_m \theta_k   \big\}, \label{steeringVecSquinted}
	\end{align}
	where $\eta_m = \frac{f_m}{f_M}$. We can see that the direction of  AB designed for $\theta_k$ is split and points to the spatial direction $\bar{\theta}_{k,m} = \eta_m \theta_k$. It is also clear that when $f_m = f_M$, we have  $\bar{\theta}_{k,m}  = \theta_k$, which leads to the definition in (\ref{steeringVec_nominal}).

	

	\begin{lemma}
		\label{lemma1}
		Let $\mathbf{a}(\bar{\theta}_m)$ and $\mathbf{a}(\theta)$ be the beam-split-corrupted and nominal steering vectors for an arbitrary direction $\theta$ and subcarrier $m\in \mathcal{M}$ as defined in (\ref{steeringVecSquinted}) and (\ref{steeringVec_nominal}), respectively. Then, $\mathbf{a}(\bar{\theta}_m)$ achieves the maximum array gain, i.e., $A_G(\theta,m) = \frac{ |\mathbf{a}^\textsf{H}(\theta)\mathbf{a}(\bar{\theta}_m)   |^2  }{N^2} $, if $\bar{\theta}_m = \eta_m \theta $.
	\end{lemma}
	
	\begin{proof}
		
		The array gain varies across the whole bandwidth as $A_G({\theta},{m}) = \frac{|\mathbf{a}^\textsf{H}(\theta)  \mathbf{a}(\bar{\theta}_m)|^2}{N^2}.$ By using (\ref{steeringVec_nominal}) and (\ref{steeringVecSquinted}), $A_G({\theta},{m})$ is 
		\begin{align}
		&A_G({\theta},{m})= \frac{1}{N^2} \left| \sum_{n_1 =1}^{N}  \sum_{n_2=1}^{N} e^{-\mathrm{j} \pi  \left( (n_1-1){\bar{\theta}_m} - (n_2-1)\frac{\lambda_M\theta}{\lambda_m}\right)    }   \right| ^2 \nonumber \\
		&= \left|\sum_{n = 0}^{N-1} \frac{e^{-\mathrm{j}2\pi n d \left( \frac{\bar{\theta}_m}{\lambda_M} - \frac{\theta}{\lambda_m}  \right)     }  }{N} \right|^2 
		=   \left|\sum_{n = 0}^{N_\mathrm{T}-1} \frac{e^{-\mathrm{j}2\pi n d \frac{(f_c \Phi_m - f_m\Phi) }{c_0}     }  }{N} \right|^2 \nonumber\\
		&= \left| \frac{1 - e^{-\mathrm{j}2\pi Nd \frac{(f_M \bar{\theta}_m - f_m\theta)}{\nu}    }}{N (1 - e^{-\mathrm{j}2\pi d\frac{(f_M \bar{\theta}_m - f_m\theta)}{\nu} ) }}   \right|^2 \hspace{-3pt}
		= \left| \frac{\sin (\pi N \mu_m )}{N\sin (\pi \gamma_m )}    \right|^2 \hspace{-3pt}  = |\xi( \mu_m )|^2, \nonumber 
		\end{align}
		where    $\mu_m = d\frac{(f_M \bar{\theta}_m - f_m\theta)}{\nu}   $. The array gain implies that most of the power is focused only on a small portion of the beamspace due to the power-focusing capability of $\xi(\cdot)$, which substantially reduces across the subcarriers as $|f_m - f_M|$ increases. Furthermore, $|\xi( \mu_m )|^2$ gives peak when $\mu_m = 0$, i.e.,  $f_M \bar{\theta}_m - f_m\theta= 0$, which yields $ \bar{\theta}_m = \eta_m \theta$. 		
	\end{proof}
	
	Above lemma indicates that the split direction $\bar{\theta}_m$ should be taken into account to achieve accurate DoA estimation performance. Finally, we define the relationship between the nominal and beam-split-corrupted steering vectors as
	\begin{align}
	\mathbf{a}(\bar{\theta}_{m,k}) = \mathbf{B}_{m,k}\mathbf{a}(\theta_k), \label{beamSquint_steeringVec1}
	\end{align}
	where $\mathbf{B}_{m,k} = \mathrm{diag}\{b_{m,k,1},\cdots, b_{m,k,N} \}$ is an $N\times N$ diagonal matrix with $b_{m,k,n} =\exp\{\mathrm{j}(n-1) \pi \Delta_{m,k}  \}  $, where  $\Delta_{m,k}$ is defined as beam-split, i.e., $\Delta_{m,k} = \bar{\theta}_{m,k} -  \theta_{k}.$
	
	\textbf{Mutual Coupling:} {Experimentally, the MC matrix is found from the inverse of the measurement transformation matrix $\mathbf{Z}_m\in \mathbb{C}^{N\times N}$, which maps the coupled voltages to the uncoupled voltages~\cite{mc_1_Gupta1983Sep,ref8_ULA_MCM,groundCalElbir} as $\mathbf{Z}_m \mathbf{V}_m = \mathbf{U}_m$. Here, $\mathbf{V}_m = \left[\mathbf{v}_m^1,\cdots, \mathbf{v}_m^\Omega  \right]\in \mathbb{C}^{N\times \Omega}$ comprises the coupled measurements collected for $\Omega \geq N$ distinct directions when all the antennas are residing in the array whereas $\mathbf{U}_m = \left[\mathbf{u}_m^1,\cdots, \mathbf{u}_m^\Omega  \right]\in \mathbb{C}^{N\times \Omega}$ is obtained with $N$ antenna elements, each of which is considered one-by-one separately while the remaining $N-1$ antennas are removed. Then, the direction-independent MC matrix is obtained as $\mathbf{C}_m = \mathbf{Z}_m^{-1}$. When the MC is DD, it should be computed for a certain angular sectors~\cite{elbir_directionDependent_MC_Elbir2017Jan}.  }

	The MC matrix of a ULA\footnote{The proposed approach is easily applicable to different array geometries, e.g., UCA/URA. Specifically, the MC matrix has circular-symmetric and block-Toeplitz structure for UCA and URA, respectively~\cite{friedlander}.} is represented by an  $L$-banded Toeplitz matrix~\cite{friedlander,elbir_DDMC_calibration_Elbir2017Oct,doaEst_MCoupling_convergence_Liao2012Jan}. 
	Define $\mathbf{C}_{m}(\theta_k)\in \mathbb{C}^{N\times N}$ as the SD and direction-dependent MC matrix with $L$ MC coefficients, i.e., $\mathbf{C}_{m}(\theta_k) = \mathrm{Toeplitz}\{1, c_{m,k,1}, \cdots, c_{m,k,L}, 0, \cdots, 0 \}$, where $\{c_{m,k,l}\}_{l = 1}^L$ are the distinct MC coefficients. Denote the actual steering vector  by $\tilde{\mathbf{a}}(\bar{\theta}_{m,k})\in \mathbb{C}^{N}$. {\color{black}The steering vector corrupted by both beam-split and MC  becomes
		\begin{align}
		{\color{black}\tilde{\mathbf{a}}(\bar{\theta}_{m,k}) = \mathbf{C}_{m}(\theta_k) \mathbf{a}(\bar{\theta}_{m,k}) =\mathbf{C}_{m}(\theta_k) \mathbf{B}_{m,k}\mathbf{a}(\theta_k), }
		\end{align}}
	for which we define the MC-corrupted steering matrix in a compact for as $\mathbf{D}_m = [\tilde{\mathbf{a}}(\bar{\theta}_{m,1}), \cdots, \tilde{\mathbf{a}}(\bar{\theta}_{m,K})] \in \mathbb{C}^{N\times K}$.

	Our goal is to estimate the DoAs of the targets $\{\theta_k\}_{k =  1}^K$ while accurately compensating for beam-split and MC.

	\section{Proposed Method}


	Define ${\mathbf{X}}_{m} = [\mathbf{x}_{m}(t_1),\cdots, \mathbf{x}_{m}(t_T)]\in \mathbb{C}^{N\times T}$ as the radar probing signal transmitted by the BS for $T$ data snapshots along the fast-time axis~\cite{mimoRadar_WidebandYu2019May}, where $\mathbb{E}\{ {\mathbf{X}}_{m} {\mathbf{X}}_{m}^\textsf{H} \} = \frac{P_\mathrm{r}T}{M N}\mathbf{I}_{N}$, for which $\mathbf{F} \mathbf{F}^\textsf{H} = 1/N$, and  $P_\mathrm{r}$ is the radar transmit power. {\color{black}Limited number of RF chains yield $N_\mathrm{RF}\times 1$ received baseband data vector. When $K \geq N_\mathrm{RF}$, it results in poorer parameter estimation~\cite{wideband_doaEst_Wideband_Friedlander1993Apr,music_Schmidt1986Mar}. To collect the full array data from $N_\mathrm{RF}$ RF chains, we follow a subarrayed approach. That is, the BS activates the antennas in a subarrayed manner to apply the $N\times N_\mathrm{RF}$ analog beamforming matrix $\widetilde{\mathbf{W}} = \left[\overline{\mathbf{W}}_1^\textsf{T},\cdots,\overline{\mathbf{W}}_J^\textsf{T}  \right]^\textsf{T}$. Then, the BS collects the received target echoes for $J = \frac{N}{N_\mathrm{RF}}$ time slots. The target DoAs remain invariant within a time slot but change across different time slots. This is reasonable for the THz system, wherein the symbol time is of the order of picoseconds~\cite{milliDegree_doa_THz_Chen2021Aug,mimoRadar_WidebandYu2019May}. Then, for the $j$-th time slot, the BS applies the combiner matrix $\mathbf{W}_j = \left[\begin{array}{c}
		\mathbf{0}_{jN_\mathrm{RF}\times N_\mathrm{RF}}\\
		\overline{\mathbf{W}}_j\\
		\mathbf{0}_{N-(j+1)N_\mathrm{RF}\times N_\mathrm{RF}} 
		\end{array}\right] \in \mathbb{C}^{N\times N_\mathrm{RF}}$, where $\overline{\mathbf{W}}_j\in \mathbb{C}^{N_\mathrm{RF}\times N_\mathrm{RF}}$ represents the $j$-th block of $\widetilde{\mathbf{W}}$ corresponding to the $j$-th subarray. }  Then, the $N_\mathrm{RF}\times T$ echo signal from the $K$ targets at the $j$-th time slot is
	\begin{align}
	\label{radarReceived}
	{\mathbf{Y}}_{m,j} = \sum_{k = 1}^K \beta_k \mathbf{W}_j^\textsf{H}\tilde{\mathbf{a}}(\bar{\theta}_{m,k}) \left[\tilde{\mathbf{a}}(\bar{\theta}_{m,k})\right]^\textsf{T} {\mathbf{X}}_{m} + \tilde{\mathbf{N}}_{m,j},
	\end{align}
	where $\beta_k\in \mathbb{C}$ denotes the reflection coefficient of the $k$-th target. $\tilde{\mathbf{N}}_{m,j} = \mathbf{W}_j^\textsf{H}\bar{\mathbf{N}}_{m,j}\in \mathbb{C}^{N_\mathrm{RF}\times T}$  is representing the noise term, where $\bar{\mathbf{N}}_{m,j} = [\bar{\mathbf{n}}_{m,j}(t_1),\cdots, \bar{\mathbf{n}}_{m,j}(t_T)]\in \mathbb{C}^{N\times T}$ with $\bar{\mathbf{n}}_{m,j}(t_i)\sim \mathcal{CN}(\mathbf{0},{\sigma}_n^2\mathbf{I}_{N})$. Denote the target steering matrix and reflection coefficients by ${\mathbf{A}}_{m,j} = \mathbf{W}_j^\textsf{H}\mathbf{D}_m\in \mathbb{C}^{N_\mathrm{RF}\times K}$ and $\boldsymbol{\Pi} = \mathrm{diag}\{\beta_1, \cdots, \beta_K \}\in \mathbb{C}^{K\times K}$, then (\ref{radarReceived}) becomes $	{\mathbf{Y}}_{m,j} = {\mathbf{A}}_{m,j}    \boldsymbol{\Pi} {\mathbf{D}}_{m}^\textsf{T}{\mathbf{X}}_{m} + \tilde{\mathbf{N}}_{m,j}.$
	Stacking all $\mathbf{Y}_{m,j}$ into a single $N\times T$ matrix leads to the overall observation matrix $\mathbf{Y}_m\in \mathbb{C}^{N\times T}$ as
	\begin{align}
	\mathbf{Y}_m &= \left[ \mathbf{Y}_{m,1}^\textsf{T}, \cdots, \mathbf{Y}_{m,J}^\textsf{T} \right]^\textsf{T}  = \mathbf{A}_m \boldsymbol{\Pi} \mathbf{D}_m^\textsf{T} \mathbf{X}_m + \tilde{\mathbf{N}}_m,
	\end{align}
	where $\mathbf{A}_m = \left[\mathbf{A}_{m,1}^\textsf{T},\cdots, \mathbf{A}_{m,J}^\textsf{T}  \right]^\textsf{T} = \mathbf{W}^\textsf{H}\mathbf{D}_m \in \mathbb{C}^{N\times K}$,   $\mathbf{W} = \left[{\mathbf{W}}_1,\cdots, {\mathbf{W}}_J  \right]\in \mathbb{C}^{N\times N}$, and $\tilde{\mathbf{N}}_m = \left[\tilde{\mathbf{N}}_{m,1}^\textsf{T},\cdots, \tilde{\mathbf{N}}_{m,J}^\textsf{T}  \right]^\textsf{T}$.
	{\color{black}We now introduce an alternating algorithm, wherein the beam-split-corrected DoA angles and the MC coefficients are estimated iteratively such that MC parameters are kept fixed while estimating the DoA angle or vice versa.  }
	
	\subsection{DoA and Beam-Split Estimation}
	In order to estimate the target directions, we invoke the wideband MUSIC algorithm~\cite{music_Schmidt1986Mar,wideband_doaEst_Wideband_Friedlander1993Apr}. Define $\mathbf{R}_m\in \mathbb{C}^{N\times N}$ as the covariance matrix of ${\mathbf{Y}}_m$, i.e., 
	\begin{align}
	\mathbf{R}_m &= \frac{1}{T}{\mathbf{Y}}_m {\mathbf{Y}}_m^\textsf{H} = \frac{1}{T} {\mathbf{A}}_m \left( \frac{P_\mathrm{r}T}{MN}\widetilde{\boldsymbol{\Pi} }\right) {\mathbf{A}}_m^\textsf{H} +  \frac{1}{T}\tilde{\mathbf{N}}_m\tilde{\mathbf{N}}_m^\textsf{H} \nonumber\\
	& \approxeq \frac{P_\mathrm{r}}{MN} {\mathbf{A}}_m \widetilde{\boldsymbol{\Pi}} {\mathbf{A}}_m^\textsf{H}  +  {\sigma}_n^2  \mathbf{I}_{{N}}, \label{R_m1}
	\end{align}
	where $\tilde{\mathbf{N}}_m\tilde{\mathbf{N}}_m^\textsf{H} \approxeq {\sigma}_n^2 T/N\mathbf{I}_{N} $ and  $\widetilde{\boldsymbol{\Pi} }\in \mathbb{C}^{K\times K} $ is defined as $	\widetilde{\boldsymbol{\Pi} } =  \boldsymbol{\Pi}\mathbf{D}^\textsf{T}\mathbf{D}^*\boldsymbol{\Pi}^*$.  	Then, the eigendecomposition of $\mathbf{R}_m$ yields $	\mathbf{R}_m = \mathbf{U}_m \boldsymbol{\Theta}_m \mathbf{U}_m^\textsf{H},$
	where $\boldsymbol{\Theta}_m\in \mathbb{C}^{N\times N}$ is a diagonal matrix composed of the eigenvalues of $\mathbf{R}_m$ in a descending order, and $\mathbf{U}_m = \left[\mathbf{U}_{m}^\mathrm{S}\hspace{2pt} \mathbf{U}_m^\mathrm{N} \right]\in \mathbb{C}^{N\times N}$ corresponds to the eigenvector matrix; $\mathbf{U}_m^\mathrm{S}\in\mathbb{C}^{N\times K}$ and $\mathbf{U}_m^\mathrm{N}\in \mathbb{C}^{N\times N-K}$ are the signal and noise subspace eigenvector matrices, respectively. The columns of $\mathbf{U}_m^\mathrm{S}$ and  ${\mathbf{A}}_m$ span the same space that is orthogonal to the eigenvectors in $\mathbf{U}_m^\mathrm{N}$ as $\| {\mathbf{U}_m^\mathrm{N}}^\textsf{H}(  \mathbf{W}^\textsf{H} \tilde{\mathbf{a}}(\bar{\theta}_{m,k}) )\|_2^2$, i.e., $	\| {\mathbf{U}_m^\mathrm{N}}^\textsf{H}( \mathbf{W}^\textsf{H} \mathbf{C}_m(\theta_{k}) \mathbf{B}_{m,k}\mathbf{a}(\theta_k) )\|_2^2 = 0,$
	for $k\in \mathcal{K}$ and $m\in \mathcal{M}$~\cite{music_Schmidt1986Mar}. Thus, given the MC matrix $\mathbf{C}_m(\theta)$, the conventional MUSIC spectra is given by 
	\begin{align}
	\label{musicSpectra2}
	P(\theta) = \sum_{m = 1}^M P_{m}(\theta),
	\end{align}
	where $	P_{m}(\theta) $ is the spectrum corresponding to the $m$-th subcarrier as
	$P_m(\theta) = \frac{1}{ f( m,k) }$, where $	f(m,k) = \mathbf{a}^\textsf{H}(\theta)\mathbf{C}_m^\textsf{H}(\theta)\mathbf{W}\mathbf{U}_m^\mathrm{N}{\mathbf{U}_m^\mathrm{N}}^\textsf{H} \mathbf{W}^\textsf{H}\mathbf{C}_m(\theta)\mathbf{a}(\theta).$
	The MUSIC spectra in (\ref{musicSpectra2}) yields $MK$ peaks, which are deviated due to beam-split (see Fig.~\ref{fig_MUSIC_spectra}a) while correct MUSIC spectra should include $K$ peaks which are aligned for $m\in \mathcal{M}$. In other words, beam-split-corrected steering vectors should be used to accurately compute the MUSIC spectra. The CREAM-MUSIC accounts for this deviation by employing \textit{beam-split-aware} steering vectors $\mathbf{e}_m(\theta)  = \mathbf{W}^\textsf{H} \mathbf{C}_m(\theta) \mathbf{a}(\bar{\theta}_m)  \in \mathbb{C}^{N}$. Then, the CREAM-MUSIC spectra is  $	\widetilde{P}(\theta) = \sum_{m = 1}^M \widetilde{P}_{m}(\theta),$ where 
	\begin{align}
	\label{musicSpectraBSA}
	\widetilde{P}_{m}(\theta) = \frac{1}{\mathbf{e}_m^\textsf{H}(\theta)\mathbf{U}_m^\mathrm{N}{\mathbf{U}_m^\mathrm{N}}^\textsf{H}\mathbf{e}_m(\theta) },
	\end{align}
	for which, the $K$ highest peaks of (\ref{musicSpectraBSA}) correspond to the estimated target DoAs $\{\hat{\theta}_k\}_{k = 1}^K$, and the beam-split is computed as $\hat{\Delta}_{m,k} = (\eta_m-1)\hat{\theta}_k$, for $m \in \mathcal{M}$, $k \in \mathcal{K}$. {\color{black}Note that combining the spectra of $M$ subcarriers results in only a single peak-search in place of separately estimating the DoAs for each subcarrier.}
	


	\subsection{Mutual Coupling Estimation}
	Define $\mathbf{c}_{m,k} = [c_{m,k,1},\cdots, c_{m,k,L}]^\textsf{T}$ as the $L\times 1$ vector MC coefficients. Then, we construct the following useful matrix-vector transformation between $\mathbf{c}_{m,k}$ and $\mathbf{C}_m(\theta_k)$, i.e., $	\mathbf{C}_m(\theta_k) = \mathbf{T}_{m,k}\mathbf{c}_{m,k},$
	where $\mathbf{T}_{m,k} = \left[\mathbf{S}_{m,1} \mathbf{a}(\bar{\theta}_{m,k}), \cdots, \mathbf{S}_{m,L}\mathbf{a}(\bar{\theta}_{m,k})  \right]\in \mathbb{C}^{N\times L}$, for which  $\mathbf{S}_{m,l}$ is an ${N\times N}$ matrix, \textcolor{black}{and it is defined for any array geometry} as $	\mathbf{S}_{m,l} = \left\{\begin{array}{ll}1, & \text{if } [\mathbf{C}_m(\theta_k)]_{i,j}= c_{m,k,l}   \\
	0, & \text{otherwise}	 \end{array}\right.$~\cite{elbir_directionDependent_MC_Elbir2017Jan}.
	Given the DoAs $\{\theta_k\}_{k = 1}^K$, we solve the following optimization problem to estimate $\mathbf{c}_{m,k}$, i.e., $	\minimize_{\mathbf{c}_{m,k}  } \hspace{3pt} \sum_{m = 1}^{M} f(m,k),
	\hspace{5pt}\subjectto  c_{m,k,1} = 1, \label{opt4c}$
	which is equivalent to $	\minimize_{\mathbf{c}_{m,k}  } \hspace{3pt} \sum_{m = 1}^{M} \mathbf{c}_{m,k}^\textsf{H} \boldsymbol{\Sigma}_{m,k} \mathbf{c}_{m,k}, \hspace{3pt}\subjectto  \mathbf{v}^\textsf{T}\mathbf{c}_{m,k} = 1, $
	where $\mathbf{v} = [1, 0, \cdots, 0]^\textsf{T}\in \mathbb{C}^{L\times 1}$ and $	\boldsymbol{\Sigma}_{m,k}\in \mathbb{C}^{L\times L}$ is $\boldsymbol{\Sigma}_{m,k} = \mathbf{T}_{m,k}^\textsf{H}\mathbf{W}\mathbf{U}_m^\mathrm{N}{\mathbf{U}_m^\mathrm{N}}^\textsf{H} \mathbf{W}^\textsf{H}\mathbf{T}_{m,k}.$
	Then, the closed-form solution for $\mathbf{c}_{m,k}$ is 
	\begin{align}
	\hat{\mathbf{c}}_{m,k} = \boldsymbol{\Sigma}_{m,k}^{-1} \mathbf{v} \left(\mathbf{v}^\textsf{T}\boldsymbol{\Sigma}_{m,k}^{-1}\mathbf{v} \right)^{-1}. \label{estC}
	\end{align}
	
	{\color{black}In Algorithm~\ref{alg}, we present the proposed alternating approach to effectively estimate the DoAs, beam-split, and MC coefficients.} Specifically, we first partition the angular search space into $S$ sectors as  $\Phi = \cup_{s = 1}^S \Psi_s$, where $\Phi = [-\frac{\pi}{2},\frac{\pi}{2}]$ for a ULA, and  $\Psi_s$ denotes the angular sector, for which the direction-dependent MC matrix $\mathbf{C}_m(\Psi_s)$ is kept fixed~\cite{elbir_directionDependent_MC_Elbir2017Jan}. Then, the estimates of the DoAs and MC coefficients are alternatingly computed until the algorithm converges for a predefined error threshold parameter $\epsilon$. {\color{black}While the alternating algorithm does not guarantee optimality, its convergence has been shown in prior works~\cite{friedlander,doaEst_MCoupling_convergence_Liao2012Jan}. Nevertheless, the proposed approach almost achieves the CRB (see Fig.~\ref{fig_DOA_RMSE_SNR}).    }
	
	{The implementation of CREAM-MUSIC is similar to the other alternating algorithms for DoA and MC estimation and beam-split compensation stage does not impose an additional constraint on the problem. The computational complexity of the proposed approach is similar to the existing techniques~\cite{elbir_DDMC_calibration_Elbir2017Oct,doaEst_MCoupling_convergence_Liao2012Jan,friedlander} except that wideband processing is involved. Therefore, the complexity order is $\mathcal{O}(M[N^3 + KL^3])$ because of eigendecomposition for DoA estimation ($\mathcal{O}(MN^3)$) and computation of the direction-dependent MC coefficients ($\mathcal{O}(MKL^3)$) for $M$ subcarriers.}

	\section{Numerical Experiments}
	We evaluate the performance of  our CREAM-MUSIC approach in comparison with the MUSIC algorithm with no calibration,  only beam-split compensation (BSC), and only MC calibration (MCC), as well as the Cram\'er-Rao bound (CRB)~\cite{elbir_THZ_CE_ArrayPerturbation_Elbir2022Aug}, in terms of root mean-squared-error (RMSE), i.e., $\mathrm{RMSE}_\theta = (\frac{1}{J_TK} \sum_{i=1}^{J_T}\sum_{k\in \mathcal{K}} || \hat{{\theta}}_{i,k}- {{\theta}}_{i,k}||_2^2 )^{1/2}$, where $\hat{{\theta}}_{i,k}$ stands for the estimated DoA  for the $i$-th experiment 	of $J_T= 100$ Monte Carlo trials. The simulation parameters are $f_c=300$ GHz, $B=30$ GHz, $M=64$, $N = 128$, $N_\mathrm{RF}=8$, $T=100$, $K=2$ and $L = 5$~\cite{elbir_directionDependent_MC_Elbir2017Jan,elbir2021JointRadarComm,milliDegree_doa_THz_Chen2021Aug}. Our CREAM-MUSIC method in Algorithm~\ref{alg} is run  approximately  for $\ell = 50$ iterations with $\epsilon = 10^{-4} $.  The DoAs are selected uniform randomly as $\tilde{\theta}_{k}\sim \mathrm{unif} [-\frac{\pi}{2},\frac{\pi}{2}]$. The AB matrix is modeled as $[\mathbf{W}]_{i,j} = \frac{1}{\sqrt{N}}e^{\mathrm{j}{\psi}}$, where ${\psi} \sim \text{unif}[-1,1]$ for $i = 1,\cdots, N$ and $j = 1,\cdots, N_\mathrm{RF}$. {\color{black}The direction-dependent MC coefficient vectors are selected as $\mathbf{c}_{m,1} = [ 0.85e^{\mathrm{j} \varphi_{m,1,1}}, 0.8e^{\mathrm{j} \varphi_{m,1,2}}, 0.4e^{\mathrm{j} \varphi_{m,1,3}}, 0.2e^{\mathrm{j} \varphi_{m,1,4}} ]^\textsf{T}$ and $\mathbf{c}_{m,2} = [ 0.9e^{\mathrm{j} \varphi_{m,2,1}}, 0.75e^{\mathrm{j} \varphi_{m,2,2}}, 0.45e^{\mathrm{j} \varphi_{m,2,3}}, 0.25e^{\mathrm{j} \varphi_{m,2,4}} ]^\textsf{T}$, where $\varphi_{k,l}\sim \mathrm{unif}[-\pi, \pi]$ for $m\in \mathcal{M}$ and  $l = 1,\cdots,L$.}

	\begin{algorithm}[t]
		\footnotesize
		\begin{algorithmic}[1] 
			\caption{ \bf CREAM-MUSIC}
			\Statex {\textbf{Input:}  $\mathbf{Y}_m$, $\mathbf{W}$, $K$, $S$, $\Phi$, $\epsilon$, $\eta_m$.} \label{alg}
			\State \textbf{Initialize:} $\ell=1$, $\mathbf{C}_m^{\ell}(\Psi_s) = \mathbf{I}_N$, $m\in \mathcal{M}$, $k\in \mathcal{K}$. 
			\State \textbf{while} not terminated \textbf{do}
			\State  \textbf{for} $m\in\mathcal{M}$ 
			\State \indent  Compute $\mathbf{R}_m$ and $\mathbf{U}_m^\mathrm{N}$ from (\ref{R_m1}).
			\State \indent  \textbf{for} $s =1, \cdots, S $ \textbf{do}
			\State \indent \indent   Compute $\widetilde{P}_{m}^{\ell}(\Psi_s)$ from (\ref{musicSpectraBSA}) using $\mathbf{C}_m^{\ell}(\Psi_s)$.
			\State \indent \textbf{end}
			\State \indent Construct $\widetilde{P}_m^{\ell} (\Phi) =  \cup_{s = 1}^S \widetilde{P}_m^{\ell}(\Psi_s)$.
			\State \textbf{end}  
			\State $\widetilde{P}^{\ell}(\Phi)  \gets \sum_{m = 1}^M \widetilde{P}_m^{\ell}(\Phi)$.
			\State Find $\{\hat{\theta}_k^{\ell}\}_{k =1}^{K}$ from the $K$ highest peaks of $\widetilde{P}^{\ell}(\Phi)$.
			\State \textbf{for} $k\in \mathcal{K}$, $m\in \mathcal{M}$
			\State \indent  $\hat{\Delta}_{m,k}^{\ell} \gets (\eta_m -1) \hat{\theta}_k^{\ell}$. 
			\State \textbf{end}
			\State Solve (\ref{estC}) for $\hat{\mathbf{C}}_{m}^{\ell}(\theta_k)$ by using $\mathbf{a}(\hat{\theta}_k^{\ell})$, $m \in \mathcal{M}$, $k\in \mathcal{K}$.
			\State \textbf{if} $  \sum_{k\in \mathcal{K}} |\hat{\theta}_k^{\ell} - \hat{\theta}_k^{\ell-1}  | \leq \epsilon $ \textbf{then} terminate  \textbf{end}
			\State $i \gets i + 1$
			\State \textbf{end}
			\Statex {\textbf{Return:} $\hat{\theta}_k$, $\hat{\Delta}_{m,k}$ and $\hat{\mathbf{C}}_m(\theta_{k})$ for $m\in \mathcal{M}$, $k \in \mathcal{K}$.}
		\end{algorithmic} 
	\end{algorithm}
	
	\begin{figure}[h]
		\centering
		{\includegraphics[draft=false,width=.95\columnwidth]{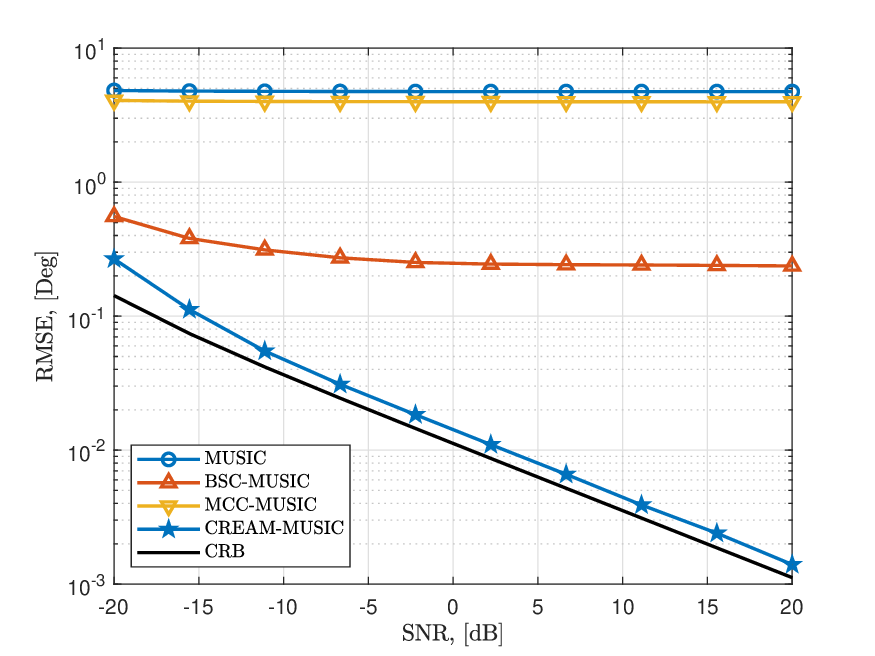} } 
		\caption{ DoA estimation RMSE vs. SNR. 
		}
		\label{fig_DOA_RMSE_SNR}
	\end{figure}

	In Fig.~\ref{fig_DOA_RMSE_SNR}, we present the DoA estimation RMSE with respect to signal-to-noise-ratio (SNR), {\color{black}which is defined as $\mathrm{SNR} = 10\log_{10}(\frac{\rho}{\sigma_n^2})$, where $\rho = \frac{P_\mathrm{r}}{MN^2} = 1$}. As it is seen, both MUSIC and MCC-MUSIC have poor performance because they do not take into account the effect of beam-split. In contrast, BSC-MUSIC exhibits lower RMSE than that of MCC-MUSIC. {\color{black}Specifically, DoA estimation error due to beam-split and MC is approximately $4^\circ$ and $0.25^\circ$, respectively.} The significance of the DoA error due to beam-split is directly related to $\eta_m$ (see Fig.\ref{fig_MUSIC_spectra}a), which causes deviations in the steering vector model (see (\ref{steeringVecSquinted})).   This clearly shows the importance of beam-split compensation. In high SNR, the performance of BSC-MUSIC maxes out due to lose of precision, while the proposed CREAM-MUSIC approach attains the CRB very closely and outperforms the  remaining methods yielding poor  precision.
	

	%
	%
	%


	\section{Summary}
	We examined the THz DoA estimation problem in the presence of beam-split and MC. While the latter has a marginal impact   ($\sim 0.25^\circ$) on DoA estimation, the former causes significant errors in the array gain and be severe ($\sim 4^\circ$). We showed that the proposed CREAM-MUSIC approach can effectively compensate both DoA errors due to beam-split and MC. Furthermore, the proposed method does not require additional hardware components, e.g., time-delayer networks for beam-split calibration. 
	

	%
	%
	%

	\newpage
	\balance
	\bibliographystyle{IEEEtran}
	\bibliography{references_122}

	%
	%
	%
	%
	%
	%
	%
	%
	

	%

\end{document}